\documentclass[11pt]{article}

\usepackage[margin=1in]{geometry}
\usepackage{amsmath,amssymb,amsthm}
\usepackage{booktabs}
\usepackage{enumitem}
\usepackage{float}
\usepackage{graphicx}
\usepackage{hyperref}
\usepackage{microtype}
\usepackage[section]{placeins}
\usepackage[numbers,sort&compress]{natbib}

\hypersetup{
  colorlinks=true,
  linkcolor=blue,
  citecolor=blue,
  urlcolor=blue
}

\newtheorem{theorem}{Theorem}
\newtheorem{lemma}{Lemma}
\newtheorem{proposition}{Proposition}
\newtheorem{definition}{Definition}
\DeclareMathOperator{\TV}{TV}

\title{Auditing Combinatorial Randomness from Finite Transcripts}
\author{
Faruk Alpay\textsuperscript{*} \quad Levent Sar{\i}oglu\\
Department of Computer Engineering, Bahcesehir University, Istanbul, Turkey\\
\texttt{faruk.alpay@bahcesehir.edu.tr}, \texttt{levent.sarioglu@bahcesehir.edu.tr}\\
\textsuperscript{*}Correspondence: \texttt{alpay@lightcap.ai}
}
\date{}

\begin{document}
\maketitle

\begin{abstract}
Public randomness services expose finite transcripts: lottery histories,
randomness-beacon pulses, and generator audit streams.  A transcript auditor
must control false alarms under the public null while retaining power against
tampering or implementation faults.  For combinatorial draws without
replacement, the transcript space has support size $\binom{m}{k}$, and arbitrary
joint-uniformity testing inherits the sparse-uniformity barrier
$\Theta(\sqrt{\binom{m}{k}}/\varepsilon^2)$.  We formulate this barrier as a
black-box auditability frontier and construct generator-agnostic geometric
audits for structured alternatives.  The battery maps draws to the hypersimplex
and evaluates exact-combinatorial null calibrated statistics:
marginal chi-square, pair maxima, serial overlap, anchored-box discrepancy, and
low-dimensional H0/MST geometry.  On 1,956 EuroMillions main-number draws from
2004-02-13 through 2026-06-19, no statistic survives false-discovery adjustment
among observed/reference sources (minimum observed/reference BH $q=0.279$).
In contrast, GPU Monte Carlo experiments with up to 300,000 exact-null
replications and 60,000 alternative replications per condition show that marginal-preserving
joint alternatives are invisible to one-dimensional marginal tests while being
detected by geometric/pair statistics.  For a 5/50 draw at $n=1,956$, a
block-cluster alternative of strength 0.04 is detected by pair maxima with
power 0.638 while marginal chi-square remains at 0.051; a band-repulsion
alternative of strength 0.08 is detected by anchored boxes with power 0.741
while marginal chi-square remains at 0.051.  The results separate marginal
goodness-of-fit from joint-structure audit power and quantify the sample sizes
at which specified low-dimensional failure modes become detectable.
\end{abstract}

\section{Introduction}

Public randomness is a security primitive when multiple parties rely on a source
whose internal entropy is not directly observable.  Lotteries, public randomness
beacons, and deployed pseudorandom generators all leave audit transcripts; the
security-relevant object is the finite transcript and the class of deviations it
can rule out.  For a draw of $k$ labels without replacement from $m$, the exact
null distribution is uniform on
\[
  \Omega_{m,k}=\{x\in\{0,1\}^m:\|x\|_1=k\},
  \qquad |\Omega_{m,k}|=\binom{m}{k}.
\]
The main EuroMillions component has
$\binom{50}{5}=2{,}118{,}760$ possible outcomes, while a 6/90 format has
$622{,}614{,}630$.  Single-coordinate count tests see only a projection of this
space.  A distribution can preserve every single-number marginal and still alter
pair probabilities, serial overlap, or the geometry of sorted draws.

The security limitation is information-theoretic.  Any black-box transcript
audit for arbitrary alternatives on $\Omega_{m,k}$ inherits Paninski's
uniformity-testing scale, so lottery-length transcripts cannot certify full
joint uniformity without restricting the adversary or fault model.  The
constructive route is to audit named structured families: hot labels, pair
couplings, serial stickiness, clustered blocks, and band repulsion.  These
families model low-dimensional tampering or implementation faults and can be
tested by exact-null calibrated geometric statistics.

The contributions are:
\begin{enumerate}[leftmargin=1.4em]
  \item A transcript-audit model for public combinatorial randomness, with audit
  advantage defined against a named alternative class.
  \item A lower-bound corollary showing that arbitrary joint-uniformity audits
  require $\Omega(\sqrt{\binom{m}{k}}/\varepsilon^2)$ samples.
  \item Generator-agnostic statistics on the draw hypersimplex:
  marginal counts, pair maxima, serial overlap, anchored-box discrepancy, and
  H0/MST geometry.  The anchored boxes are calibrated against the exact
  combinatorial model.
  \item CPU and CUDA implementations, with GPU experiments on an NVIDIA RTX 5090
  and an NVIDIA A100 SXM4.  The GPU experiments add
  marginal-preserving alternatives that directly expose the blind spot of
  one-dimensional audits.
\end{enumerate}

\section{Related Work}

Classical RNG testing has a mature empirical tradition, including TestU01
\citep{lecuyer2007testu01}.  Lattice, spectral, and discrepancy methods connect
random-number generation to geometry and quasi-Monte Carlo theory
\citep{niederreiter1992}.  These tests are especially powerful for structured
linear generators.  The present setting is black-box: the input is a finite
sequence of combinatorial draws, with no access to a generator recurrence.
For cryptographic random-bit generation, NIST SP 800-22 provides statistical
tests for output streams and SP 800-90B specifies validation requirements for
entropy sources \citep{nist80022,nist80090b}.  Statistical transcript auditing
addresses observable distributional faults; it is distinct from computational
unpredictability and from implementation-level entropy claims.

Distribution testing gives the relevant lower-bound scale.  Paninski's
coincidence test and matching lower bound show that uniformity over a support of
size $N$ requires $\Theta(\sqrt{N}/\varepsilon^2)$ samples for constant success
probability in total variation distance \citep{paninski2008}.  The implication
for draws without replacement is immediate once $N=\binom{m}{k}$.  Canonne's
survey gives the broader property-testing context for these sample-complexity
barriers \citep{canonne2020}.  The applied statistical treatment of coincidences
has a parallel tradition in audit problems, including the birthday-style methods
of \citet{diaconis1989coincidences}.

Computational geometry enters through discrepancy and topological summaries.
Exact star discrepancy is computationally hard in high dimension
\citep{gnewuch2009}, motivating finite, calibrated families of anchored boxes
rather than exact high-dimensional optimization.  Topological data analysis
provides statistically meaningful summaries such as persistence landscapes
\citep{bubenik2015}, with standard algorithmic foundations in computational
topology \citep{edelsbrunner2010,ghrist2008}.  Random geometric complexes give
null theory for geometric point clouds \citep{bobrowski2018}.  GPU persistent
homology systems such as Ripser++ demonstrate that these computations can be
scaled \citep{zhang2020ripserpp}.  The H0 statistic used below is exactly the
Euclidean minimum-spanning-tree edge spectrum in the projected draw cloud.

Public randomness beacons are becoming more auditable.  The NIST beacon format
specifies signed, timestamped pulses and interoperability goals
\citep{kelsey2019nistir8213}.  CURBy is a public beacon based on traceable
quantum randomness \citep{kavuri2025curby,curbyweb}; our data script cached
16,000 bytes from 250 completed CURBy rounds during the 2026-06-20 run.
EuroMillions draws were obtained from the public EuroMillions API repository
\citep{euromillionsapi}; the local cache contains 1,956 main-number draws.

\section{Security Model and Auditability Frontier}

\begin{definition}[Transcript auditor]
Let $U$ be the uniform law on $\Omega_{m,k}$.  A randomized transcript auditor is
a map $A:\Omega_{m,k}^n\to\{0,1\}$, where $A=1$ denotes rejection.  It has level
$\alpha$ if $\Pr_{U^n}(A=1)\le\alpha$.  Its audit advantage against a source law
$P$ is
\[
  \operatorname{Adv}_A(P)
  =\Pr_{P^n}(A=1)-\Pr_{U^n}(A=1).
\]
\end{definition}

\begin{definition}[Arbitrary-alternative radius]
Let $N=\binom{m}{k}$.  The sparse-uniformity scale for arbitrary alternatives is
\[
  \varepsilon_{\mathrm{arb}}(m,k,n)
  = \left(\frac{\sqrt{N}}{n}\right)^{1/2}.
\]
It is the total-variation separation obtained by solving
$n\asymp \sqrt{N}/\varepsilon^2$ for $\varepsilon$.
\end{definition}

\begin{theorem}[Black-box audit lower bound]
Fix constants $\alpha,\beta\in(0,1/2)$ and let
\[
  \mathcal{P}_\varepsilon=\{P:\TV(P,U)\ge\varepsilon\}.
\]
If a level-$\alpha$ auditor satisfies
$\inf_{P\in\mathcal{P}_\varepsilon}\Pr_{P^n}(A=1)\ge1-\beta$, then
\[
  n=\Omega_{\alpha,\beta}\!\left(
    \frac{\sqrt{\binom{m}{k}}}{\varepsilon^2}
  \right).
\]
This order is achievable up to constants by sparse-uniformity testing.
\end{theorem}

\begin{proof}
Set $N=\binom{m}{k}$ and identify each draw with one of the $N$ atoms in
$\Omega_{m,k}$.  The statement is
Paninski's minimax uniformity-testing lower and upper scale applied to this
induced discrete distribution \citep{paninski2008}.
\end{proof}

For 5/50, $\sqrt{N}=1{,}456$ and $n=1{,}956$ gives the constant-one separation
scale
$\varepsilon \approx \sqrt{\sqrt{N}/n}=0.863$.  For 6/49 with $n=2{,}300$, the
same calculation is $1.275$; for 6/90 with $n=1{,}000$, it is $4.995$, outside
the range of total variation.  No level-controlled black-box auditor can offer
uniform power over all separated source laws at these transcript lengths.

\begin{proposition}[Marginal non-identifiability]
For $2\le k\le m-2$, there are non-uniform distributions on $\Omega_{m,k}$ with
exactly uniform single-coordinate marginals.
\end{proposition}

\begin{proof}
Choose four distinct labels $a,b,c,d$.  For $S\in\Omega_{m,k}$ define
\[
 h(S)=\mathbf{1}\{a,b\subset S\}+\mathbf{1}\{c,d\subset S\}
      -\mathbf{1}\{a,c\subset S\}-\mathbf{1}\{b,d\subset S\}.
\]
The sum of $h(S)$ over all $S$ is zero.  Conditioning on any fixed label also
gives zero: for labels in $\{a,b,c,d\}$ the positive and negative terms cancel,
and for any other label all four terms have the same count
$\binom{m-3}{k-3}$, with the convention that this count is zero when $k=2$.
Hence, for sufficiently small $|\theta|$,
\[
  P_\theta(S)=\binom{m}{k}^{-1}\{1+\theta h(S)\}
\]
is a valid non-uniform distribution whose single-coordinate marginals match the
uniform law.  Pair probabilities are changed whenever $\theta\ne0$.
\end{proof}

\begin{theorem}[Structured witness audit]
Let $\mathcal{F}=\{f_1,\ldots,f_L\}$ with
$f_j:\Omega_{m,k}\to[0,1]$, let $\mu_j=\mathbb{E}_U f_j(X)$, and define
\[
  D_{\mathcal F}
  =\max_{1\le j\le L}
   \left|\frac{1}{n}\sum_{i=1}^n f_j(X_i)-\mu_j\right|.
\]
The threshold
\[
  \tau_\alpha=\sqrt{\frac{\log(2L/\alpha)}{2n}}
\]
defines a level-$\alpha$ auditor.  If
$\Delta_{\mathcal F}(P)=\max_j|\mathbb{E}_P f_j-\mu_j|$ satisfies
\[
  \Delta_{\mathcal F}(P)\ge
  \tau_\alpha+\sqrt{\frac{\log(2/\beta)}{2n}},
\]
then the rejection probability under $P$ is at least $1-\beta$.
\end{theorem}

\begin{proof}
Under $U^n$, Hoeffding's inequality and a union bound over the $L$ witnesses give
$\Pr_{U^n}(D_{\mathcal F}>\tau_\alpha)\le\alpha$.  Let $f_{j^\star}$ attain
the witness gap under $P$.  A second Hoeffding bound places its empirical mean
within $\sqrt{\log(2/\beta)/(2n)}$ of its $P$-expectation with probability at
least $1-\beta$; the displayed gap then forces rejection.
\end{proof}

Anchored boxes use witnesses $f_t(y)=\mathbf{1}\{y\le t\}$.  Pair indicators
and bounded serial-overlap functions give analogous auditors.  Their sample
requirements depend on the witness gap and $\log L$, rather than directly on
the full support size $\binom{m}{k}$.

\begin{lemma}[Exact Monte Carlo calibration]
Let $T$ be any statistic whose large values are more extreme, and let
$T_0^{(1)},\ldots,T_0^{(B)}$ be i.i.d. exact-null replications generated at the
same $(m,k,n)$ as the observed sample.  The randomized-rank p-value
\[
  \hat p=\frac{1+\sum_{b=1}^{B}\mathbf{1}\{T_0^{(b)}\ge T_{\mathrm{obs}}\}}
              {B+1}
\]
satisfies $\Pr_{H_0}(\hat p\le \alpha)\le \alpha$ for every $\alpha\in[0,1]$.
\end{lemma}

\begin{proof}
Under $H_0$, $T_{\mathrm{obs}},T_0^{(1)},\ldots,T_0^{(B)}$ are exchangeable.
The rank of $T_{\mathrm{obs}}$ among the $B+1$ values is therefore uniform up to
ties; using the upper-tail count with the additive one is the standard
conservative tie handling.
\end{proof}

\begin{lemma}[Null moments for pair and serial statistics]
For a fixed unordered pair $\{i,j\}$, the per-draw inclusion probability is
\[
  p_2=\Pr(i,j\in X)=\frac{k(k-1)}{m(m-1)}.
\]
Thus its count over $n$ independent draws is $\mathrm{Binomial}(n,p_2)$
marginally.  For two consecutive independent draws,
$|X_t\cap X_{t-1}|$ is hypergeometric with
\[
  \mathbb{E}|X_t\cap X_{t-1}|=\frac{k^2}{m}, \qquad
  \mathrm{Var}(|X_t\cap X_{t-1}|)
  =k\frac{k}{m}\left(1-\frac{k}{m}\right)\frac{m-k}{m-1}.
\]
\end{lemma}

\begin{proof}
The pair probability follows by sampling two named labels without replacement.
For the overlap, condition on $X_{t-1}$.  The next draw samples $k$ labels from a
population of size $m$ containing $k$ marked labels, giving the displayed
hypergeometric moments \citep{serfling1974,chvatal1979}.
\end{proof}

\section{Methods}

\paragraph{Exact null.}
Every null sample is drawn without replacement from $\{1,\ldots,m\}$ and sorted.
Let $Y=(Y_{(1)},\ldots,Y_{(k)})$ be the sorted zero-indexed labels.  For an
anchored threshold vector $t=(t_1,\ldots,t_k)$, the exact null probability is
\[
  q_t=\binom{m}{k}^{-1}
      \#\{0\le y_1<\cdots<y_k<m:\; y_j\le t_j \;\forall j\}.
\]
The CUDA implementation evaluates this count by dynamic programming:
\[
  C(r,a)=\sum_{v=a}^{\min(t_r,m-k+r)} C(r+1,v+1), \qquad C(k,a)=1,
\]
with $q_t=C(0,0)/\binom{m}{k}$.  This avoids a continuous-uniform
approximation for the sorted simplex.

\paragraph{Statistics.}
We report five statistic families:
\begin{itemize}[leftmargin=1.4em]
  \item marginal chi-square,
  $\sum_{i=1}^m(C_i-nk/m)^2/(nk/m)$;
  \item pair max-z,
  $\max_{i<j}|(C_{ij}-np_2)/\sqrt{np_2(1-p_2)}|$;
  \item serial overlap-z, based on the hypergeometric moments above;
  \item anchored-box discrepancy,
  $\max_{t\in\mathcal{T}}|n^{-1}\sum_{\ell=1}^n\mathbf{1}\{Y_\ell\le t\}-q_t|$;
  \item H0/MST geometry, obtained by projecting incidence vectors with a fixed
  Gaussian map and summarizing the Euclidean minimum-spanning-tree edge spectrum.
\end{itemize}
The maximum statistics are calibrated by exact-null Monte Carlo, so dependence
among pairs or boxes is kept in the reference distribution.

\paragraph{Alternative families.}
We use five controlled alternatives:
\begin{itemize}[leftmargin=1.4em]
  \item hot-marginal: a low-index subset receives multiplicative sampling weight;
  \item pair-injection: a fixed pair is injected with probability $\rho$;
  \item block-cluster: with probability $\rho$, choose one of $m/k$ blocks and
  return that whole block; single-number marginals remain uniform;
  \item band-repulsion: with probability $\rho$, choose one number from each of
  $k$ equal bands; single-number marginals remain uniform;
  \item serial-stickiness: with probability $\rho$, copy one label from the
  previous draw and fill the rest uniformly.
\end{itemize}

\paragraph{Calibration.}
All p-values and power values are Monte Carlo calibrated at the same sample size
as the tested source.  The local source audit used 300 null replications.  The
GPU experiments used 50,000--100,000 null replications per $n$ and
10,000--20,000 alternative replications per condition in the first pass, with
larger calibrated replications for the final GPU tables.  Discrete statistics
use conservative upper-tail rejection against the calibrated critical value to
avoid size inflation from ties.

\section{Results}

\subsection{Observed and Reference Sources}

The local source audit uses EuroMillions main-number draws from 2004-02-13
through 2026-06-19, plus PCG64, MT19937, a weak LCG, controlled alternatives, and
a small CURBy cache.  Table~\ref{tab:source-pvalues} reports calibrated p-values.
For the observed EuroMillions source, the minimum p-value is 0.106.  Among the
observed/reference sources, the minimum Benjamini--Hochberg q-value is 0.279.
The EuroMillions trace is inside the calibrated reference range of this battery.

\begin{table}[!htbp]
\centering
\small
\resizebox{\linewidth}{!}{\begin{tabular}{lrrrrrrr}
\toprule
statistic & marginal\_chi2 & pair\_max\_z & serial\_overlap\_z & anchored\_box & projection\_energy & mst\_cv & h0\_total\_persistence \\
source &  &  &  &  &  &  &  \\
\midrule
curby\_qrng\_beacon & 0.030 & 0.475 & 0.339 & 0.020 & 0.342 & 0.641 & 0.801 \\
euromillions\_main & 0.163 & 0.944 & 0.834 & 0.515 & 0.106 & 0.282 & 0.522 \\
lcg\_minstd & 0.189 & 0.751 & 0.774 & 0.040 & 0.578 & 0.156 & 0.807 \\
mt19937\_reference & 0.013 & 0.120 & 0.854 & 0.944 & 0.661 & 0.316 & 0.425 \\
pair\_bias\_0.12 & 0.003 & 0.003 & 0.439 & 0.003 & 0.811 & 0.080 & 0.093 \\
pcg64\_reference & 0.721 & 0.751 & 0.970 & 0.764 & 0.080 & 0.794 & 0.123 \\
sticky\_0.16 & 0.468 & 0.751 & 0.003 & 0.369 & 0.306 & 0.914 & 0.269 \\
weighted\_hot\_0.45 & 0.003 & 0.003 & 0.003 & 0.003 & 0.924 & 0.312 & 0.797 \\
\bottomrule
\end{tabular}
}
\caption{Calibrated p-values from the local source audit.  The CURBy row is
limited by the 16,000-byte cache available during the run ($n=280$ converted
5/50 draws).  EuroMillions has $n=1,956$ main-number draws.}
\label{tab:source-pvalues}
\end{table}
\FloatBarrier

\subsection{Marginal-Preserving Alternatives}

At $n=1,956$, the two marginal-preserving alternatives keep marginal chi-square
near nominal size while joint/geometric tests gain power
(Figure~\ref{fig:marginal-blind}).  Under block clustering at strength 0.04,
pair max-z reaches 0.638 power while marginal chi-square is 0.051.  Under band
repulsion at strength 0.08, anchored boxes reach 0.741 power while marginal
chi-square is 0.051.

\begin{figure}[H]
\centering
\includegraphics[width=\linewidth]{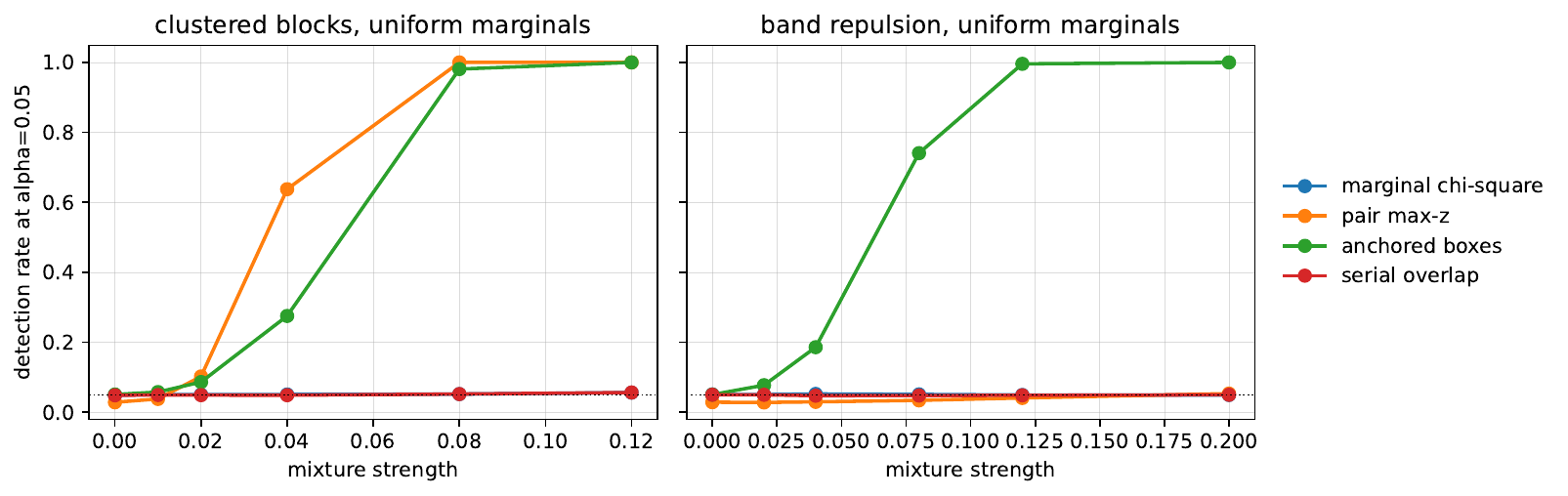}
\caption{Marginal-preserving alternatives at $n=1,956$ for 5/50 draws.  The
single-number marginal test remains at nominal size, while pair and geometric
statistics detect joint structure.}
\label{fig:marginal-blind}
\end{figure}
\FloatBarrier

\subsection{Sample-Size Frontier}

Figure~\ref{fig:sample-frontier} fixes the alternative strength at 0.04 and
varies the number of draws.  The transition is gradual: for block clustering,
pair max-z power is 0.638 at $n=1,956$ and 0.999 at $n=5,000$, while marginal
chi-square stays near 0.05.  For band repulsion, anchored-box power is 0.186 at
$n=1,956$, 0.489 at $n=5,000$, 0.882 at $n=10,000$, 0.971 at $n=12,000$, and
1.000 at $n=20,000$.  The same structured deviation can therefore have low
audit advantage at one transcript length and near-unit detection probability at
a larger one.

\begin{figure}[!htbp]
\centering
\includegraphics[width=\linewidth]{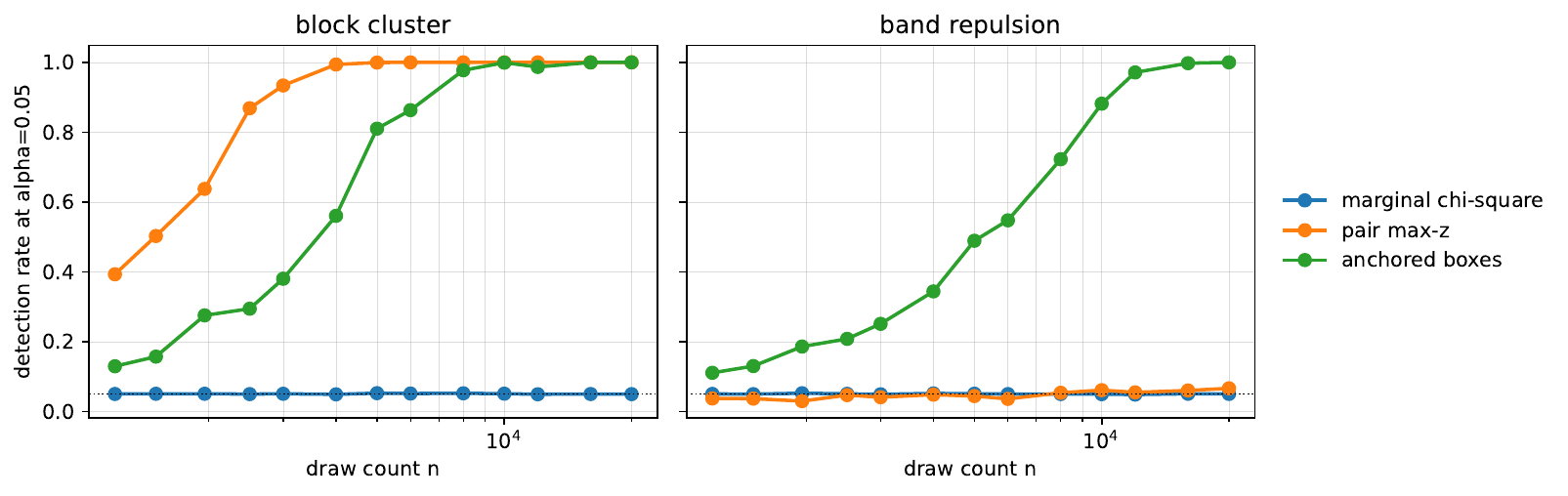}
\caption{Sample-size frontier for marginal-preserving alternatives at strength
0.04 in 5/50 draws.  The x-axis is logarithmic.}
\label{fig:sample-frontier}
\end{figure}
\FloatBarrier

\subsection{Format Scale}

The fixed 6/54 run and the A100 6/90 run repeat the GPU audit for larger
supports, including $C(90,6)=622{,}614{,}630$.  At $n=5,000$, 6/90
block-cluster strength 0.04 is detected by pair max-z with power 1.000 and by
anchored boxes with power 0.688, while marginal chi-square is 0.049.  For 6/90
band-repulsion strength 0.04, anchored boxes reach 0.651 power, while marginal
chi-square is 0.048 and pair max-z is 0.051.  The fixed 6/54 run gives the same
pattern at $n=5,000$: block clustering has pair max-z power 1.000 and anchored
box power 0.606, while band repulsion has anchored-box power 0.712 and marginal
chi-square remains 0.053.  Figure~\ref{fig:format-scale} compares 5/50, 6/54,
and 6/90 surfaces.

\begin{figure}[!htbp]
\centering
\includegraphics[width=\linewidth]{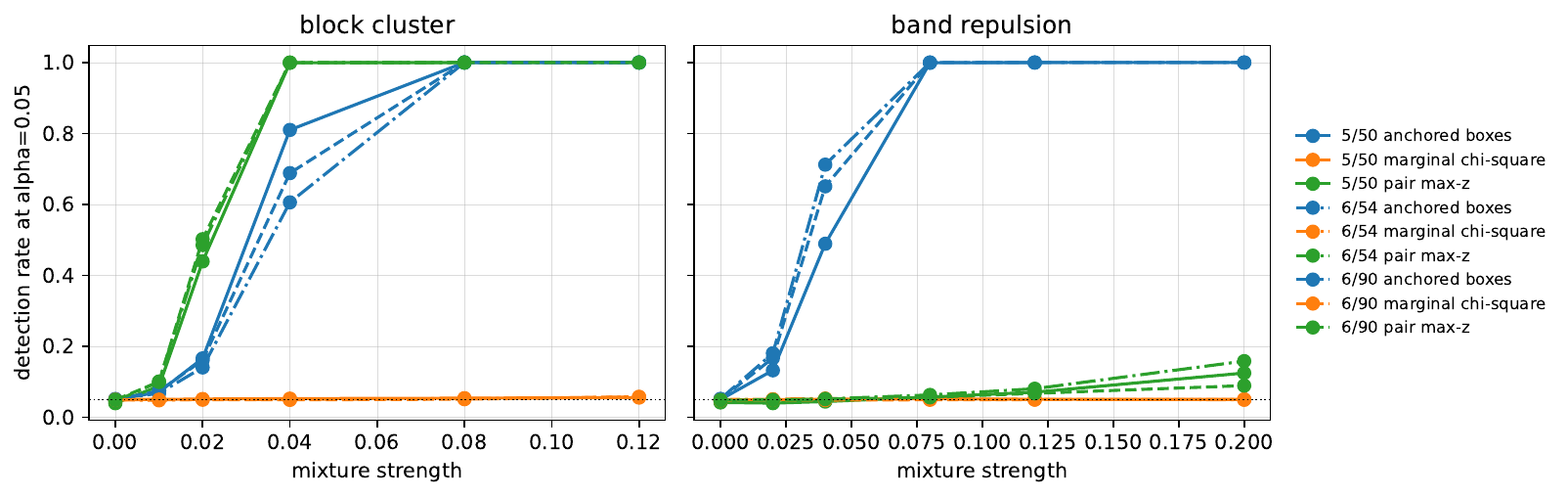}
\caption{Format-scale comparison at $n=5,000$.  Solid lines are 5/50,
dash-dotted lines are 6/54, and dashed lines are 6/90.  Marginal-preserving
alternatives keep marginal chi-square near nominal size across formats.}
\label{fig:format-scale}
\end{figure}
\FloatBarrier

\section{Discussion}

The experiments separate three regimes.  First, arbitrary joint uniformity is
outside the information budget of real lottery histories.  Second, marginal-only
audits can be exactly blind to meaningful joint alternatives.  Third, geometric
and pair statistics have useful power against stated alternatives, especially as
the audit sample grows.

Each reported value is attached to a source, null model, alternative class,
statistic, calibration budget, and power curve.  These elements define the scope
of the finite-transcript security claim and can be reproduced from the archived
code and CSV files.

\paragraph{Limitations.}
The CPU H0/MST statistic is deliberately modest; full H1/H2 persistent homology
would require an additional Ripser++-style pipeline.  The CURBy cache used in the
source table is small because the public endpoint was slow during the run; it is
included as a low-volume reference-source smoke test.  An early 6/54 exploratory
run exposed an anchored-box implementation issue and was discarded.  The
archived 6/54 CSV, main tables, and figures use the corrected implementation.

\section{Reproducibility}

All code, cached public data, generated figures, final GPU CSV files, and result
tables are included in the ancillary archive.  A SHA-256 artifact manifest
records the packaged files.  The main commands are:
\begin{verbatim}
python3 scripts/fetch_data.py
python3 scripts/run_experiments.py
nvcc -O3 -std=c++17 scripts/gpu_combinatorial_audit.cu -lcurand \
  -o build/gpu_combinatorial_audit
python3 scripts/analyze_gpu_results.py
\end{verbatim}
The GPU CSV files record device label, $n$, alternative, strength, statistic,
replication count, critical value, null mean/standard deviation, wall-clock
seconds, box count, and seed.

\section{Conclusion}

Full joint uniformity over $\Omega_{m,k}$ has the
$\sqrt{\binom{m}{k}}$ sparse-testing barrier.  Structured witness families avoid
direct dependence on the full support and require
$O((\log L+\log(1/\beta))/\Delta^2)$ samples for witness gap $\Delta$.
Exact-null calibrated marginal, pair, serial, and geometric auditors instantiate
this security tradeoff for finite combinatorial transcripts.

\bibliographystyle{plainnat}
\bibliography{refs}

@article{paninski2008,
  author = {Paninski, Liam},
  title = {A Coincidence-Based Test for Uniformity Given Very Sparsely Sampled Discrete Data},
  journal = {IEEE Transactions on Information Theory},
  volume = {54},
  number = {10},
  pages = {4750--4755},
  year = {2008},
  doi = {10.1109/TIT.2008.928987}
}

@article{canonne2020,
  author = {Canonne, Cl{\'e}ment L.},
  title = {A Survey on Distribution Testing: Your Data Is Big. But Is It Blue?},
  journal = {Theory of Computing},
  note = {Graduate Surveys 9},
  pages = {1--100},
  year = {2020},
  doi = {10.4086/toc.gs.2020.009},
  url = {https://theoryofcomputing.org/articles/gs009/}
}

@article{lecuyer2007testu01,
  author = {L'Ecuyer, Pierre and Simard, Richard},
  title = {{TestU01}: A {C} Library for Empirical Testing of Random Number Generators},
  journal = {ACM Transactions on Mathematical Software},
  volume = {33},
  number = {4},
  pages = {22},
  year = {2007},
  doi = {10.1145/1268776.1268777}
}

@techreport{nist80022,
  author = {Bassham, Lawrence E. and Rukhin, Andrew L. and Soto, Juan and Nechvatal, James R. and Smid, Miles E. and Barker, Elaine B. and Leigh, Stefan D. and Levenson, Mark and Vangel, Mark and Banks, David L. and Heckert, Nathanael Alan and Dray, James F. and Vo, San},
  title = {A Statistical Test Suite for Random and Pseudorandom Number Generators for Cryptographic Applications},
  institution = {National Institute of Standards and Technology},
  type = {Special Publication},
  number = {800-22 Revision 1a},
  year = {2010},
  doi = {10.6028/NIST.SP.800-22r1a}
}

@techreport{nist80090b,
  author = {Turan, Meltem S{\"o}nmez and Barker, Elaine and Kelsey, John and McKay, Kerry A. and Baish, Mary L. and Boyle, Mike},
  title = {Recommendation for the Entropy Sources Used for Random Bit Generation},
  institution = {National Institute of Standards and Technology},
  type = {Special Publication},
  number = {800-90B},
  year = {2018},
  doi = {10.6028/NIST.SP.800-90B}
}

@book{niederreiter1992,
  author = {Niederreiter, Harald},
  title = {Random Number Generation and Quasi-Monte Carlo Methods},
  publisher = {SIAM},
  address = {Philadelphia},
  year = {1992},
  doi = {10.1137/1.9781611970081}
}

@article{diaconis1989coincidences,
  author = {Diaconis, Persi and Mosteller, Frederick},
  title = {Methods for Studying Coincidences},
  journal = {Journal of the American Statistical Association},
  volume = {84},
  number = {408},
  pages = {853--861},
  year = {1989},
  doi = {10.1080/01621459.1989.10478847}
}

@article{serfling1974,
  author = {Serfling, Robert J.},
  title = {Probability Inequalities for the Sum in Sampling without Replacement},
  journal = {The Annals of Statistics},
  volume = {2},
  number = {1},
  pages = {39--48},
  year = {1974},
  doi = {10.1214/aos/1176342611}
}

@article{chvatal1979,
  author = {Chvatal, Vasek},
  title = {The Tail of the Hypergeometric Distribution},
  journal = {Discrete Mathematics},
  volume = {25},
  number = {3},
  pages = {285--287},
  year = {1979},
  doi = {10.1016/0012-365X(79)90084-0}
}

@article{bubenik2015,
  author = {Bubenik, Peter},
  title = {Statistical Topological Data Analysis Using Persistence Landscapes},
  journal = {Journal of Machine Learning Research},
  volume = {16},
  number = {3},
  pages = {77--102},
  year = {2015},
  url = {https://jmlr.org/papers/v16/bubenik15a.html}
}

@book{edelsbrunner2010,
  author = {Edelsbrunner, Herbert and Harer, John},
  title = {Computational Topology: An Introduction},
  publisher = {American Mathematical Society},
  address = {Providence, RI},
  year = {2010},
  isbn = {9780821849255}
}

@article{ghrist2008,
  author = {Ghrist, Robert},
  title = {Barcodes: The Persistent Topology of Data},
  journal = {Bulletin of the American Mathematical Society},
  volume = {45},
  number = {1},
  pages = {61--75},
  year = {2008},
  doi = {10.1090/S0273-0979-07-01191-3}
}

@article{bobrowski2018,
  author = {Bobrowski, Omer and Kahle, Matthew},
  title = {Topology of Random Geometric Complexes: A Survey},
  journal = {Journal of Applied and Computational Topology},
  volume = {1},
  number = {3--4},
  pages = {331--364},
  year = {2018},
  doi = {10.1007/s41468-017-0010-0}
}

@article{gnewuch2009,
  author = {Gnewuch, Michael and Srivastav, Anand and Winzen, Carola},
  title = {Finding Optimal Volume Subintervals with $k$ Points and Calculating the Star Discrepancy Are {NP}-Hard Problems},
  journal = {Journal of Complexity},
  volume = {25},
  number = {2},
  pages = {115--127},
  year = {2009},
  doi = {10.1016/j.jco.2008.07.001}
}

@inproceedings{zhang2020ripserpp,
  author = {Zhang, Simon and Xiao, Mengbai and Wang, Hao},
  title = {{GPU}-Accelerated Computation of {Vietoris--Rips} Persistence Barcodes},
  booktitle = {36th International Symposium on Computational Geometry (SoCG 2020)},
  series = {Leibniz International Proceedings in Informatics (LIPIcs)},
  volume = {164},
  pages = {70:1--70:17},
  publisher = {Schloss Dagstuhl -- Leibniz-Zentrum f{\"u}r Informatik},
  year = {2020},
  doi = {10.4230/LIPIcs.SoCG.2020.70}
}

@article{kavuri2025curby,
  author = {Kavuri, Gautam A. and Palfree, Jasper and Reddy, Dileep V. and Zhang, Yanbao and Bienfang, Joshua C. and Mazurek, Michael D. and others},
  title = {Traceable Random Numbers from a Non-Local Quantum Advantage},
  journal = {Nature},
  year = {2025},
  doi = {10.1038/s41586-025-09054-3}
}

@techreport{kelsey2019nistir8213,
  author = {Kelsey, John and Brand{\~a}o, Lu{\'i}s T. A. N. and Peralta, Rene and Booth, Harold},
  title = {A Reference for Randomness Beacons: Format and Protocol Version 2},
  institution = {National Institute of Standards and Technology},
  type = {Draft NISTIR},
  number = {8213},
  year = {2019},
  doi = {10.6028/NIST.IR.8213-draft}
}

@misc{euromillionsapi,
  author = {Mealha, Pedro},
  title = {Euromillions Public {API}},
  year = {2026},
  howpublished = {\url{https://github.com/pedro-mealha/euromillions-api}},
  note = {Accessed 2026-06-20}
}

@misc{curbyweb,
  author = {{University of Colorado Boulder}},
  title = {{CURBy}: {CU} Randomness Beacon},
  year = {2026},
  howpublished = {\url{https://random.colorado.edu/}},
  note = {Accessed 2026-06-20}
}

\end{document}